\documentclass[12pt]{article}

\usepackage[a4paper,margin=1in]{geometry}
\usepackage{amsmath,amssymb,amsthm}
\usepackage{mathtools}
\usepackage{authblk}
\usepackage{hyperref}
\usepackage{enumitem}
\usepackage{setspace}
\usepackage{natbib}
\usepackage{tikz}
\usetikzlibrary{positioning,arrows.meta}
\usepackage{float}
\usepackage{caption}
\DeclareMathOperator*{\argmax}{arg\,max}
\newtheorem{theorem}{Theorem}

\newtheorem{proposition}{Proposition}

\newtheorem{definition}{Definition}
\newtheorem{assumption}{Assumption}

\title{\textbf{Meta-Bayesian Nash Equilibrium: Existence via Kakutani's Fixed Point Theorem}}

\author[1,*]{Madjid Eshaghi Gordji}
\author[2]{Esmaiel Abounoori}
\author[1]{Mohamadali Berahman}

\affil[1]{Faculty of Mathematics, Statistics and Computer Science, Semnan University, Semnan, Iran}
\affil[2]{Faculty of Economics, Semnan University, Semnan, Iran}

\affil[*]{Corresponding author: \href{mailto:meshaghi@semnan.ac.ir}{meshaghi@semnan.ac.ir}}

\date{}

\begin{document}

\maketitle

\noindent
\textbf{Email addresses:}\\
Madjid Eshaghi Gordji: \href{mailto:meshaghi@semnan.ac.ir}{meshaghi@semnan.ac.ir}\\
Esmaiel Abounoori: \href{mailto:esmaiel.abounoori@gmail.com}{esmaiel.abounoori@gmail.com}\\
Mohamadali Berahman: \href{mailto:mohamadali\_berahman@semnan.ac.ir}{mohamadali\_berahman@semnan.ac.ir}

\vspace{0.5cm}

\vspace{0.5cm}

\vspace{0.5cm}

\begin{abstract}
We extend the concept of meta-Nash equilibrium, recently introduced by \citet{EshaghiBagha2026} for complete-information games, to settings with incomplete information. We define a \emph{meta-Bayesian Nash equilibrium} as a profile of type-dependent mixed meta-actions and an environmental move such that no player type can profitably deviate and the environment cannot improve its mixed move. For each transformed game, meta-payoffs are defined by the unique Bayesian Nash equilibrium of that game. Using Kakutani's fixed-point theorem, we prove existence under finiteness of types, meta-actions, and transformations, together with the assumption that each transformed game admits a unique Bayesian Nash equilibrium. Three illustrative examples --- adaptive subsidies with private costs, cybersecurity protocol selection, and platform rule formation --- show that private information at the meta-level is essential and that our framework subsumes both classical Bayesian games and complete-information meta-games.
\end{abstract}



\section{Model}

For each transformed game,
\[
T(G_0) = \bigl(N, (A_i), (\Theta_i), p, (u_i^T)\bigr),
\]
where each
\[
u_i^T:A\times\Theta\to\mathbb{R}
\]
is continuous in actions. This restriction is purely technical: it keeps the domain of the fixed-point argument constant; more general transformations, e.g., action deletions, are discussed in the conclusion.

\subsection{Meta-Actions and Environment}

Each player \(i\) has a finite set of meta-actions \(X_i\). The environment is modelled as an additional player, player \(0\), with a finite set of moves \(E\). The environment's payoff function is
\[
W:X\times E\times\Theta\to\mathbb{R},
\]
where
\[
X=\prod_i X_i.
\]
Because all sets are finite, no continuity assumptions are needed.

A mixed meta-strategy for player \(i\) is a type-dependent distribution:
\[
m_i:\Theta_i \to \Delta(X_i).
\]
The environment's mixed meta-strategy is
\[
m_0\in\Delta(E),
\]
where we assume the environment has no private information for simplicity; an extension is straightforward.

Define
\[
M_i = \prod_{\theta_i\in\Theta_i} \Delta(X_i),
\]
the set of all mixed meta-strategies for player \(i\). Each \(M_i\) is a compact convex set, as it is a product of simplices. The space of all mixed meta-strategies is
\[
M = \left(\prod_{i=1}^n M_i\right) \times \Delta(E),
\]
which is also compact and convex.

\subsection{Transformation Rule and Induced Probabilities}

A transformation rule is a function
\[
\Psi: X\times E\times\Theta \to \mathcal{T}.
\]
For each pure meta-action profile \((x,e)\) and type profile \(\theta\), \(\Psi\) determines which transformation is implemented. When players use mixed meta-strategies \(m\), the probability of obtaining a particular transformation \(T\) given \(\theta\) is
\[
\pi(T\mid m,\theta)
:=
\sum_{(x,e)\in X\times E}
\mathbf{1}_{\{\Psi(x,e,\theta)=T\}}
\left(\prod_{i=1}^n m_i(x_i\mid\theta_i)\right)m_0(e).
\]
Because \(X_i\) and \(E\) are finite, each \(\pi(T\mid m,\theta)\) is a multilinear, hence continuous, function of \(m\).

\subsection{Unique Within-Game Equilibrium}

\begin{assumption}\label{ass:unique}
For every transformation \(T\in\mathcal{T}\), the transformed game \(T(G_0)\) has a unique Bayesian Nash equilibrium, denoted by
\[
\sigma(T)=(\sigma_1(T),\dots,\sigma_n(T)).
\]
\end{assumption}

This assumption allows us to define meta-payoffs without an exogenous selection rule. For each \(T\) and each type profile \(\theta\), define the ex-post equilibrium payoff:
\[
\bar u_i(T,\theta)
:=
\mathbb{E}_{a\sim\sigma(T)(\theta)} u_i^T(a,\theta).
\]
The interim meta-payoff of player \(i\) of type \(\theta_i\) under meta-strategy profile \(m\) is
\[
V_i(m\mid\theta_i)
=
\sum_{\theta_{-i}} p(\theta_{-i}\mid\theta_i)
\sum_{T\in\mathcal{T}}
\pi(T\mid m,\theta)\,\bar u_i(T,\theta).
\]
The environment's expected payoff is
\[
V_0(m)
=
\sum_{\theta\in\Theta} p(\theta)
\sum_{(x,e)\in X\times E}
\left(\prod_{i=1}^n m_i(x_i\mid\theta_i)\right)
m_0(e)\,W(x,e,\theta).
\]

\section{Meta-Bayesian Nash Equilibrium}
\label{sec:equilibrium}

\begin{definition}[Meta-Bayesian Nash equilibrium]
A profile
\[
m^* = (m_1^*,\dots,m_n^*,m_0^*)\in M
\]
is a \emph{meta-Bayesian Nash equilibrium} if for every player \(i\), every type \(\theta_i\in\Theta_i\), and every \(\mu\in\Delta(X_i)\),
\[
V_i(m^*\mid\theta_i)
\ge
V_i\bigl((\mu \text{ at } \theta_i,\; m_i^*(\cdot\mid\theta_i') \text{ for } \theta_i'\neq\theta_i,\; m_{-i}^*,m_0^*)\mid\theta_i\bigr),
\]
and for the environment,
\[
V_0(m^*)
\ge
V_0(m_1^*,\dots,m_n^*,m_0)
\qquad
\forall m_0\in\Delta(E).
\]
\end{definition}

In words, no player type can profitably deviate by changing his or her meta-action only at that type, and the environment cannot profitably change its mixed move.

\section{Existence Theorem and Proof}
\label{sec:existence}

\begin{theorem}[Existence of Meta-Bayesian Nash Equilibrium]
Assume:
\begin{enumerate}
\item For each player \(i\), the type space \(\Theta_i\) and the meta-action set \(X_i\) are finite; the environment move set \(E\) and the transformation set \(\mathcal{T}\) are finite.
\item Assumption~\ref{ass:unique} holds.
\end{enumerate}
Then a meta-Bayesian Nash equilibrium exists.
\end{theorem}

\section{Introduction}

Since \citet{Nash1951}, non-cooperative game theory has traditionally assumed that the underlying game is fixed. However, many real-world strategic situations involve a deeper layer of interaction in which the game itself may be modified through strategic behaviour, institutional intervention, or environmental forces. Early insights into this broader perspective appear in the work of \citet{Schelling1960}, who emphasised focal points and coordination mechanisms, while \citet{HarsanyiSelten1988} developed a systematic theory of equilibrium selection. In institutional economics, \citet{North1990} analysed the persistence and evolution of institutional structures over time.
Related work on the emergence and stability of conventions in strategic environments was developed by \citet{Young1993}.
More recently, \citet{EshaghiBagha2026} introduced the concept of \emph{meta-Nash equilibrium} together with the notion of status-quo inertia in deterministic environments. Relatedly, \citet{EshaghiAbounooriBerahman2026} studied equilibrium persistence and institutional transition under status-quo inertia, emphasizing how endogenous modifications of the strategic environment may be necessary to overcome equilibrium lock-in and institutional persistence.

The present paper extends that framework to incomplete-information settings in which agents possess private information that affects not only strategic actions within a game, but also preferences regarding which game should ultimately be played. In this setting, strategic interaction takes place simultaneously at two levels: players choose actions inside a game, while also participating in a higher-order process that determines how the game itself is transformed.

To focus specifically on endogenous game selection rather than equilibrium selection inside transformed games, we impose a uniqueness assumption on the Bayesian Nash equilibrium associated with each transformed game. We define a \emph{meta-Bayesian Nash equilibrium} in which players first choose type-dependent mixed meta-actions and an environmental move is selected. Together, these determine a transformation
\[
\Psi
\]
of the base Bayesian game. Each transformed game admits a unique Bayesian Nash equilibrium, and meta-payoffs are defined as the expected equilibrium payoffs generated within the transformed game.

Using Kakutani's fixed-point theorem \citep{Kakutani1941}, we prove the existence of meta-Bayesian Nash equilibrium under finiteness assumptions on type spaces, meta-actions, and transformations.Our existence argument also relies on continuity and upper hemicontinuity properties of correspondences in finite-dimensional topological spaces; see \citet{Berge1963} for related foundational results. The proposed framework contains several important special cases. When the transformation set is a singleton, the model reduces to the classical Bayesian Nash equilibrium framework studied in standard game theory \citep{FudenbergTirole1991}. When type spaces are singletons, the model reduces to the deterministic meta-Nash equilibrium framework introduced by \citet{EshaghiBagha2026}.

Three illustrative examples --- adaptive subsidies with private costs, cybersecurity protocol selection, and platform rule formation --- demonstrate that private information at the meta-level is essential for understanding endogenous institutional and strategic transformations.

The paper is organised as follows. Section~\ref{sec:framework} introduces the general framework. Section~\ref{sec:equilibrium} defines meta-Bayesian Nash equilibrium. Section~\ref{sec:existence} establishes the existence theorem. Section~\ref{sec:examples} presents illustrative examples, and Section~\ref{sec:conclusion} concludes.

\section{Framework}
\label{sec:framework}

\subsection{Base Bayesian Game}

A Bayesian game is a tuple
\[
G_0 =
\bigl(
N,
(A_i)_{i\in N},
(\Theta_i)_{i\in N},
p,
(u_i)_{i\in N}
\bigr),
\]
where
\[
N=\{1,\dots,n\}
\]
is a finite set of players, \(A_i\) is either a compact convex subset of a Euclidean space or a finite action set (in the latter case we work with mixed strategies), \(\Theta_i\) is a finite type space for player \(i\), and
\[
p:\prod_j \Theta_j \to [0,1]
\]
is a common prior distribution over type profiles.

For each player \(i\), the payoff function is
\[
u_i:A\times\Theta\to\mathbb{R},
\]
where
\[
A=\prod_i A_i,
\qquad
\Theta=\prod_i\Theta_i.
\]
We assume that each payoff function \(u_i\) is continuous in actions.

A behavioural strategy for player \(i\) is a mapping
\[
\sigma_i:\Theta_i\to\Delta(A_i),
\]
where \(\Delta(A_i)\) denotes the set of probability distributions over \(A_i\). Since each type space is finite, measurability issues are trivial.

Given a strategy profile
\[
\sigma=(\sigma_1,\dots,\sigma_n),
\]
the interim expected payoff of player \(i\) with type \(\theta_i\) is
\[
U_i(\sigma \mid \theta_i)
=
\mathbb{E}_{\theta_{-i}\sim p(\cdot\mid\theta_i)}
\left[
\mathbb{E}_{a\sim\sigma(\theta)}
u_i(a,\theta)
\right].
\]

A Bayesian Nash equilibrium is a strategy profile
\[
\sigma^*
\]
such that for every player \(i\), every type
\[
\theta_i\in\Theta_i,
\]
and every alternative behavioural strategy
\[
\sigma_i':\Theta_i\to\Delta(A_i),
\]
we have
\[
U_i(\sigma^*\mid\theta_i)
\ge
U_i\bigl((\sigma_i',\sigma_{-i}^*)\mid\theta_i\bigr).
\]

\subsection{Transformations}

Let \(\mathcal{T}\) be a finite set of transformation operators. Following \citet{EshaghiBagha2026}, a transformation
\[
T\in\mathcal{T}
\]
maps a Bayesian game to another Bayesian game. For the existence theorem, we restrict attention to transformations that leave the sets of players, action sets, type spaces, and the prior unchanged, and only modify the payoff functions.

\begin{proof}

We apply Kakutani's fixed-point theorem \citep{Kakutani1941} to the space
\[
M.
\]
The set \(M\) is a compact convex subset of a finite-dimensional Euclidean space because each \(\Delta(X_i)\) and \(\Delta(E)\) is a simplex.

\medskip

\noindent\textbf{Lemma.}

\begin{itemize}
\item
For each player \(i\) and each type \(\theta_i\), the meta-payoff
\[
V_i(\cdot\mid\theta_i)
\]
is continuous on \(M\).

\item
For fixed \((m_{-i},m_0)\) and fixed
\[
m_i(\cdot\mid\theta_i')
\]
for
\[
\theta_i'\neq\theta_i,
\]
the map
\[
\mu
\mapsto
V_i\bigl(
(\mu \text{ at } \theta_i,\;
m_i(\cdot\mid\theta_i')
\text{ for }
\theta_i'\neq\theta_i,\;
m_{-i},m_0)
\mid\theta_i
\bigr)
\]
is affine in
\[
\mu\in\Delta(X_i).
\]

\item
For fixed \((m_1,\dots,m_n)\), the map
\[
\nu\mapsto V_0(m_1,\dots,m_n,\nu)
\]
is affine in
\[
\nu\in\Delta(E),
\]
and
\[
V_0
\]
is continuous on \(M\).
\end{itemize}

\medskip

\noindent\textit{Proof of Lemma.}

Because \(X_i\), \(E\), and \(\Theta\) are finite,
\[
\pi(T\mid m,\theta)
\]
is a multilinear polynomial in the mixed meta-strategies. The expectations over \(\theta_{-i}\) and \(\theta\) are linear operations, while
\[
\bar u_i(T,\theta)
\]
and
\[
W(x,e,\theta)
\]
are constants. Hence each \(V_i\) and \(V_0\) is a finite linear combination of multilinear functions, which implies continuity and the stated affinity properties.
\hfill\(\square\)

\medskip

For each player \(i\) and each type \(\theta_i\), define the interim best-reply correspondence. Holding fixed player \(i\)'s mixed actions at all types other than \(\theta_i\), together with \(m_{-i}\) and \(m_0\), define
\[
B_{i,\theta_i}(m)
=
\argmax_{\mu\in\Delta(X_i)}
V_i\bigl(
(\mu \text{ at } \theta_i,\;
m_i(\cdot\mid\theta_i')
\text{ for }
\theta_i'\neq\theta_i,\;
m_{-i},m_0)
\mid\theta_i
\bigr).
\]

Since \(V_i\) is affine in \(\mu\) for fixed remaining components, the maximum is attained and the set of maximisers is convex. By Berge's Maximum Theorem, using continuity of the objective function and compactness of the feasible set,
\[
B_{i,\theta_i}
\]
is upper hemicontinuous in \(m\).

For the environment, define
\[
B_0(m_1,\dots,m_n)
=
\argmax_{\nu\in\Delta(E)}
V_0(m_1,\dots,m_n,\nu).
\]

Affinity in \(\nu\) together with continuity in \(m\) implies that
\[
B_0
\]
is non-empty, convex-valued, and upper hemicontinuous.

Now define the product correspondence
\[
B:M\rightrightarrows M
\]
by
\[
B(m)
=
\left(
\prod_{i,\theta_i}
B_{i,\theta_i}(m)
\right)
\times
B_0(m_1,\dots,m_n).
\]

For every
\[
m\in M,
\]
the set
\[
B(m)
\]
is non-empty and convex. Since each component correspondence is upper hemicontinuous, the product correspondence \(B\) is also upper hemicontinuous.

Therefore, Kakutani's fixed-point theorem yields a point
\[
m^*\in M
\]
such that
\[
m^*\in B(m^*).
\]

By construction, \(m^*\) satisfies the no-deviation conditions for all players, all types, and the environment. Hence,
\[
m^*
\]
is a meta-Bayesian Nash equilibrium.
\hfill\(\square\)

\end{proof}

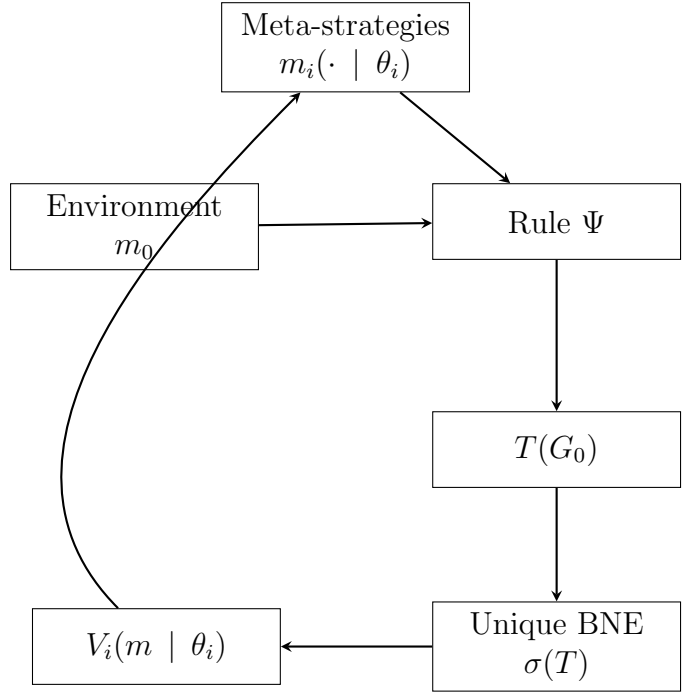
\begin{figure}[htbp]
\centering

\begin{tikzpicture}[
    node distance=1.5cm,
    auto,
    >=stealth,
    block/.style={
        rectangle,
        draw,
        text width=3cm,
        align=center,
        minimum height=1cm
    },
    arrow/.style={
        ->,
        thick
    }
]

\node[block] (meta)
{Meta-strategies\\
$m_i(\cdot\mid\theta_i)$};

\node[block, below left=1.2cm and -0.5cm of meta] (env)
{Environment\\
$m_0$};

\node[block, below right=1.2cm and -0.5cm of meta] (trans)
{Rule $\Psi$};

\node[block, below=2cm of trans] (game)
{$T(G_0)$};

\node[block, below=1.5cm of game] (equil)
{Unique BNE\\
$\sigma(T)$};

\node[block, left=2cm of equil] (payoff)
{$V_i(m\mid\theta_i)$};

\draw[arrow] (meta) -- (trans);
\draw[arrow] (env) -- (trans);
\draw[arrow] (trans) -- (game);
\draw[arrow] (game) -- (equil);
\draw[arrow] (equil) -- (payoff);

\draw[arrow]
(payoff)
to[out=135,in=225]
(meta);

\end{tikzpicture}

\caption{
Structure of the meta-Bayesian game.
Players select type-dependent mixed meta-actions and the environment selects a mixed move.
The transformation rule $\Psi$ produces a distribution over transformations $T$.
The transformed game is played, and its unique Bayesian Nash equilibrium determines meta-payoffs.
}

\label{fig:metagame}

\end{figure}

\section{Illustrative Examples}
\label{sec:examples}

Three illustrative examples demonstrate both the necessity of the framework and its applicability to contemporary economic and policy problems.

\subsection{Adaptive Subsidy Competition}

Consider two firms competing in a market where production costs are privately known. Before market competition takes place, firms can support different subsidy mechanisms through meta-actions. The government, acting as the environment, selects a subsidy regime depending on the meta-action profile and partially observed economic conditions.

Each subsidy regime transforms the payoff structure of the underlying Bayesian game. Firms therefore strategically influence not only market behaviour, but also the institutional environment in which competition occurs. The resulting equilibrium depends simultaneously on private information and endogenous game transformation.

\subsection{Cybersecurity Protocol Selection}

Consider several firms connected through a shared digital infrastructure. Each firm privately knows its vulnerability level and operational costs. Before operational interaction occurs, firms select meta-actions corresponding to preferred cybersecurity standards or protocols.

An external regulator or coordinating authority determines which protocol becomes mandatory. Different protocols induce different transformed Bayesian games because security costs and exposure risks change across protocols. The framework captures the interaction between private information, strategic institutional choice, and equilibrium behaviour under uncertainty.

\subsection{Platform Rule Formation}

Consider a digital platform in which users and moderators possess private information regarding preferences, moderation tolerance, or reputational concerns. Before participation in the platform economy, agents strategically support alternative moderation or governance rules.

The platform environment selects a governance structure according to the aggregate meta-actions. Different governance rules generate different transformed games affecting participation incentives, visibility, and strategic interaction among users. In this setting, endogenous rule formation becomes part of the strategic equilibrium itself.

\subsection{Adaptive Subsidy Competition with Private Costs}

Two firms compete in prices in a differentiated product market. Each firm has a private cost type
\[
\theta_i\in\{L,H\},
\qquad
0<c_L<c_H,
\]
with independent prior
\[
\Pr(\theta_i=H)=q.
\]
The within-game action set is
\[
A_i=[0,\bar p],
\]
where \(p_i\in A_i\) denotes the price chosen by firm \(i\). For a transformation \(T\), the payoff of firm \(i\) is
\[
u_i^T(p_i,p_j,\theta_i)
=
\bigl(p_i-c(\theta_i)+s_i(T)\bigr)
\bigl(a-bp_i+\gamma p_j\bigr),
\]
where \(s_i(T)\) is the subsidy received by firm \(i\).

In the adaptive tax/subsidy example, let
\[
\mathcal{T}=\{T^0,T^1,T^2,T^S\}
\]
denote the set of possible policy regimes, where \(T^0\) is the neutral regime with no subsidy to either firm, \(T^1\) grants a subsidy only to firm \(1\), \(T^2\) grants a subsidy only to firm \(2\), and \(T^S\) grants the same subsidy to both firms. Accordingly, the subsidy component entering firm \(i\)'s payoff is defined directly by
\[
s_i(T^0)=0
\]
for both firms,
\[
(s_1(T^1),s_2(T^1))=(\bar s,0),
\]
\[
(s_1(T^2),s_2(T^2))=(0,\bar s),
\]
and
\[
(s_1(T^S),s_2(T^S))=(\bar s,\bar s).
\]
Thus, under any realized transformation \(T\), each firm's payoff in the induced price-competition game is given by
\[
u_i^T(p_i,p_j,\theta_i)
=
\bigl(p_i-c(\theta_i)+s_i(T)\bigr)D_i(p_i,p_j),
\]
with the policy regime determining whether the firm receives no subsidy or the amount \(\bar s\) before the downstream Bertrand interaction.

Each firm has a meta-action
\[
x_i\in\{0,1\},
\]
\[
x_i\in\{0,1\},
\]
where \(x_i=1\) represents lobbying and \(x_i=0\) represents no lobbying. The environment, interpreted as a regulator, has a move
\[
e\in\{N,R\},
\]
where \(N\) means neutral and \(R\) means responsive. Under \(N\), the regulator always selects \(T^0\). Under \(R\), the implemented transformation depends on the firms' lobbying decisions.

The transformation rule is
\[
\Psi(x_1,x_2,e)
=
\begin{cases}
T^0, & e=N,\\
T^1, & e=R,\ x_1=1,\ x_2=0,\\
T^2, & e=R,\ x_1=0,\ x_2=1,\\
T^S, & e=R,\ x_1=x_2=1,\\
T^0, & \text{otherwise}.
\end{cases}
\]

For each transformation \(T\), assume that the transformed Bayesian game has a unique Bayesian Nash equilibrium \(\sigma(T)\). Let
\[
\bar u_i(T,\theta)
\]
denote the corresponding ex-post equilibrium payoff. The interim meta-payoff of firm \(i\), including lobbying cost \(\kappa_i(\theta_i)\ge 0\), is
\[
V_i(m\mid\theta_i)
=
\sum_{\theta_{-i}}p(\theta_{-i}\mid\theta_i)
\sum_{T\in\mathcal{T}}
\pi(T\mid m,\theta)\bar u_i(T,\theta)
-
\kappa_i(\theta_i)m_i(1\mid\theta_i).
\]

\begin{proposition}
Assume that for every \(\theta_j\),
\[
\bar u_i(T^1,(L,\theta_j))
-
\bar u_i(T^0,(L,\theta_j))
>
\bar u_i(T^1,(H,\theta_j))
-
\bar u_i(T^0,(H,\theta_j)),
\]
and that
\[
\kappa_i(L)\le \kappa_i(H).
\]
Then in any meta-Bayesian Nash equilibrium,
\[
m_i(1\mid L)\ge m_i(1\mid H).
\]
\end{proposition}

\begin{proof}[Proof sketch]
Let
\[
\Delta_i(\theta_i)
\]
be the expected gain from switching \(x_i\) from \(0\) to \(1\), holding the other strategies fixed. Because the transformation rule \(\Psi\) makes either \(T^1\) or \(T^S\) more likely when \(x_i=1\), the payoff-difference assumption implies
\[
\Delta_i(L)\ge \Delta_i(H).
\]
In equilibrium,
\[
m_i(1\mid\theta_i)=1
\]
if
\[
\Delta_i(\theta_i)>0,
\]
while
\[
m_i(1\mid\theta_i)=0
\]
if
\[
\Delta_i(\theta_i)<0.
\]
If
\[
\Delta_i(\theta_i)=0,
\]
then any mixed value is optimal. Therefore, the monotonicity of \(\Delta_i\) implies
\[
m_i(1\mid L)\ge m_i(1\mid H).
\]
\end{proof}

\subsection{Cybersecurity Protocol Selection with Private Vulnerabilities}

Two players must later choose a protocol
\[
a_i\in\{A,B\}.
\]
Each player has a private type
\[
\theta_i\in\{S,W\},
\]
where \(S\) denotes a secure type and \(W\) denotes a weak type. Under transformation \(T\), the payoff of player \(i\) is
\[
u_i^T(a_i,a_j,\theta_i)
=
\begin{cases}
b_T-c_T(a_i,\theta_i), & a_i=a_j,\\
-\ell_T(\theta_i)-c_T(a_i,\theta_i), & a_i\neq a_j,
\end{cases}
\]
where
\[
\ell_T(W)>\ell_T(S).
\]

Let
\[
\mathcal{T}
=
\{T^{\mathrm{open}},T^{\mathrm{strict}},T^{\mathrm{legacy}}\}.
\]
The meta-actions are
\[
x_i\in\{o,s\},
\]
where \(o\) denotes support for an open protocol and \(s\) denotes support for a strict protocol. The environment has moves
\[
e\in\{\mathrm{statusquo},\mathrm{reform}\}.
\]
The transformation rule is analogous to the previous example.

To fit the uniqueness assumption, we suppose that the protocol game is perturbed, for example by a small noise term, so that each transformed game
\[
T(G_0)
\]
has a unique Bayesian Nash equilibrium. Since weak types face larger mismatch losses, they have a stronger preference for the strict protocol. Therefore, in equilibrium,
\[
m_i(s\mid W)>m_i(s\mid S).
\]

\subsection{Platform Rule Formation with Private Seller Quality}

There are two sellers,
\[
i=1,2,
\]
with private quality
\[
\theta_i\in\{q_L,q_H\},
\qquad
0<q_L<q_H.
\]
After a rule \(T\) is selected, each seller chooses entry and price:
\[
A_i=\{\mathrm{enter},\mathrm{stay\ out}\}\times[0,\bar p].
\]

If seller \(i\) enters and sets price \(p_i\), then the payoff is
\[
u_i^T(a,\theta)
=
\mathbf{1}_{\{\mathrm{entry}\}}
\left[
(1-r_T)p_iD_i^T(p,\theta)-f_T(\theta_i)
\right],
\]
where
\[
D_i^T(p,\theta)
=
\delta_T(\theta_i)-\beta p_i+\gamma p_j,
\]
and
\[
\delta_T(q_H)>\delta_T(q_L).
\]

Let
\[
\mathcal{T}
=
\{T^L,T^S,T^Q\},
\]
where \(T^L\) denotes a lax rule, \(T^S\) denotes a strict rule, and \(T^Q\) denotes a quality-based rule. The meta-actions are
\[
x_i\in\{\mathrm{Lax},\mathrm{Strict}\}.
\]
The environment has moves
\[
e\in\{\mathrm{Growth},\mathrm{Quality}\}.
\]

The transformation rule is as follows:
\[
\Psi(x_1,x_2,e)
=
\begin{cases}
T^L, & e=\mathrm{Growth}
\text{ and } x_1=x_2=\mathrm{Lax},\\
T^S, & e=\mathrm{Quality}
\text{ and } x_1=x_2=\mathrm{Strict},\\
T^Q, & \text{otherwise}.
\end{cases}
\]

Assume that each transformed game
\[
T(G_0)
\]
has a unique Bayesian Nash equilibrium. Since high-quality sellers benefit more from strict rules, their equilibrium support for strict governance is stronger. Hence,
\[
m_i(\mathrm{Strict}\mid q_H)
>
m_i(\mathrm{Strict}\mid q_L).
\]

\section{Conclusion}
\label{sec:conclusion}

We have extended the meta-Nash equilibrium framework introduced by \citet{EshaghiBagha2026} to environments with incomplete information. The proposed notion of meta-Bayesian Nash equilibrium captures situations in which players strategically influence not only actions within a Bayesian game, but also the determination of which Bayesian game is ultimately played. In this setting, private information affects both meta-level strategic choices and equilibrium behaviour inside transformed games.

Under finite type spaces, finite meta-action sets, and a finite family of payoff-preserving transformations, we established the existence of meta-Bayesian Nash equilibrium using Kakutani's fixed-point theorem \citep{Kakutani1941}. The proof relied on continuity, convexity, and upper hemicontinuity properties of the induced best-response correspondence on the space of mixed meta-strategies.

The framework contains several important benchmark models as special cases. When the transformation set
\[
\mathcal{T}
\]
is a singleton, the model reduces to the classical Bayesian Nash equilibrium framework. When type spaces are degenerate, the model collapses to the deterministic meta-Nash equilibrium framework developed in \citet{EshaghiBagha2026}.

The examples of adaptive subsidy competition, cybersecurity protocol selection, and platform rule formation demonstrate that endogenous transformation of games under incomplete information naturally arises in economics, regulation, institutional design, and digital governance. In all these settings, private information influences not only equilibrium actions, but also strategic preferences over institutional structures themselves.

Several extensions remain open for future research. One direction is to relax the uniqueness assumption on within-game Bayesian Nash equilibria by introducing equilibrium selection correspondences, correlated equilibrium selections, or stochastic refinement mechanisms. Another direction is to allow transformations that modify not only payoff functions, but also action sets, type spaces, timing structures, or information structures. Dynamic and repeated versions of meta-Bayesian games may also provide a foundation for studying long-run institutional evolution and strategic adaptation under uncertainty.

\section*{Data Availability Statement}

No datasets were generated or analysed during the current study.

\section*{Conflict of Interest}

The authors declare that they have no conflict of interest.

\bibliographystyle{plainnat}
\bibliography{references}

\end{document}